\documentclass[11pt]{article}

\usepackage{amsmath,amssymb,amsthm}
\usepackage{fullpage}
\usepackage{authblk}
\usepackage[position=top]{subfig}
\usepackage{graphicx}

\usepackage{tikz}
\usetikzlibrary{arrows}
\usetikzlibrary{positioning, fit}
\usetikzlibrary{shapes,snakes}
\usetikzlibrary{decorations.pathmorphing}
\usetikzlibrary{calc,through,intersections}
\usetikzlibrary{patterns}
\usetikzlibrary{trees}
\usetikzlibrary{arrows.meta}
\tikzset{snake it/.style={decorate, decoration=snake}}
\tikzset{arc/.style = {->,> = latex', line width=.75pt}}

\title{Around the diameter of AT-free graphs}
\author[1,2]{Guillaume Ducoffe}
\affil[1]{\small National Institute for Research and Development in Informatics, Romania}
\affil[2]{\small University of Bucharest, Romania}
\date{}

\newtheorem{lemma}{Lemma}
\newtheorem{theorem}{Theorem}
\newtheorem{corollary}{Corollary}
\newtheorem{proposition}{Proposition}

\begin{document}

\maketitle

\begin{abstract}
A graph algorithm is {\em truly subquadratic} if it runs in ${\cal O}(m^b)$ time on connected $m$-edge graphs, for some positive $b < 2$. Roditty and Vassilevska Williams ({\it STOC}'13) proved that under plausible complexity assumptions, there is no truly subquadratic algorithm for computing the diameter of general graphs. In this work, we present positive and negative results on the existence of such algorithms for computing the diameter on some special graph classes. Specifically, three vertices in a graph form an {\em asteroidal triple} (AT) if between any two of them there exists a path that avoids the closed neighbourhood of the third one. We call a graph {\em AT-free} if it does not contain an AT. 
%We first study the fine-grained complexity of diameter computation within AT-free graphs: for which we show how to compute all the eccentricities in truly subquadratic ${\cal O}(m^{3/2})$ time. 
We first prove that for all $m$-edge AT-free graphs, one can compute all the eccentricities in truly subquadratic ${\cal O}(m^{3/2})$ time. 
Then, we extend our study to several subclasses of chordal graphs -- all of them generalizing interval graphs in various ways --, as an attempt to understand which of the properties of AT-free graphs, or natural generalizations of the latter, can help in the design of fast algorithms for the diameter problem on broader graph classes. 
%Our second main result is an ${\cal O}(m\log^2{n})$-time randomized algorithm for computing the diameter with high probability on $n$-vertex $m$-edge chordal graphs with a bounded asteroidal number. Furthermore, there is a truly subquadratic algorithm for computing the diameter of chordal dominating pair graphs. Finally, for all chordal graphs with a dominating diametral path (resp., with a dominating triple), if the diameter is at least $4$ (resp., at least $7$), then there is a linear-time algorithm for computing a diametral pair. However, already for split graphs with a dominating edge, under plausible complexity assumptions, there is no truly subquadratic algorithm for deciding whether the diameter is either $2$ or $3$.
For instance, for all chordal graphs with a dominating shortest path, there is a linear-time algorithm for computing a diametral pair if the diameter is at least four. However, already for split graphs with a dominating edge, under plausible complexity assumptions, there is no truly subquadratic algorithm for deciding whether the diameter is either $2$ or $3$.
\end{abstract}

\section{Introduction}\label{sec:intro}

For any undefined graph terminology, see~\cite{BoM08}. All graphs studied in this paper are finite, simple and connected. Given a graph $G=(V,E)$, let $n = |V|$ be its order and $m = |E|$ be its size. Note that, since we assume $G$ to be connected, $m \geq n-1$. For a vertex $u \in V$, let $N_G(u) = \{ v \in V \mid uv \in E \}$ and $N_G[u] = \{u\} \cup N_G(u)$ be, respectively, the open and closed neighbourhoods of $u$. The distance between two vertices $u,v \in V$ is equal to the minimum number of edges on a $uv$-path, and it is denoted by $dist_G(u,v)$. The maximum such distance between a fixed vertex $u$ and all other vertices is called its eccentricity, denoted by $e_G(u) = \max_{v \in V} dist_G(u,v)$. We sometimes omit the subscript if the graph $G$ is clear from the context. Finally, the {\em diameter} of a graph $G$ is equal to $diam(G) = \max_{u,v \in V} dist(u,v) = \max_{u \in V} e(u)$. The purpose of this note is to study the complexity of diameter computation within some graph classes. On general graphs, the best known algorithms for this problem run in ${\cal O}(nm)$ time and in ${\cal O}(n^{\omega+o(1)})$ time, respectively, where $\omega < 2.3729$ is the exponent of square matrix multiplication~\cite{Sei95}. In particular, both algorithms run in $\Omega(m^2)$ time on sparse graphs, where by sparse we mean that $m \leq c \cdot n$ for some universal constant $c$.

Improving the quadratic running time for diameter computation -- in the number $m$ of edges -- is an important research problem,  both in theory and in practice.
Since any algorithm for this problem must run in $\Omega(m)$ time, some authors have considered whether we can compute the diameter in linear time or in quasi linear time on certain graph classes. Notably, this is the case for the {\em interval graphs} -- {\it i.e.}, the intersection graphs of intervals on the real line --~\cite{Ola90}. All graph classes that are considered in this work are superclasses of the interval graphs.
For general graphs, less than a decade ago~\cite{RoV13}, Roditty and Vassilevska Williams gave convincing evidence that the complexity of diameter computation {\em cannot} be improved by much. Specifically, the Strong Exponential-Time Hypothesis (SETH) says that for any $\varepsilon > 0$, there exists a $k$ such that $k$-{\sc SAT} on $n$ variables cannot be solved in ${\cal O}((2-\varepsilon)^n)$ time~\cite{ImP01}. The Orthogonal-Vector problem ({\sc OV}) takes as input two families $A$ and $B$ of $n$ sets over some universe $C$, and it asks whether there exist $a \in A, b \in B$ s.t. $a \cap b = \emptyset$. Some older works have named this problem {\sc Disjoint Sets}~\cite{ChD92}. The following result is due to Williams:

\begin{theorem}[\cite{Wil05}]\label{thm:ov}
Under {\sc SETH}, for any $\varepsilon > 0$, there exists a constant $c > 0$ such that we cannot solve {\sc OV} in ${\cal O}(n^{2-\varepsilon})$ time, even if $|C| \leq c \cdot \log{n}$.
\end{theorem}

In what follows, by {\em truly subquadratic} we mean a running time in ${\cal O}(N^{2-\varepsilon})$, for some $\varepsilon > 0$, where $N$ denotes the size of the input (for connected graphs, $N \approx m$). This above Theorem~\ref{thm:ov} was used in order to prove that many classic problems that can be solved in polynomial time do not admit a truly subquadratic algorithm ({\it e.g.}, see~\cite{VaW18} for a survey). In particular, as far as we are concerned in this note, under SETH there is no truly subquadratic algorithm for computing the diameter, even on sparse graphs~\cite{AVW16}.
This negative result has motivated a long line of papers, with some trying to characterize the graph classes for which there {\em does} exist a truly subquadratic algorithm for the diameter problem.
We refer to~\cite{AVW16,BCT17,Cab18,CDP19} for recent relevant work in this area.
In particular, several authors have studied whether some important properties of interval graphs could imply on their own fast diameter computation algorithms. This is the case for Helly graphs~\cite{DuD19+}, and graphs of bounded distance VC-dimension~\cite{DHV20}, but {\em not} for chordal graphs~\cite{BCH16}. 

We recall that a graph is called AT-free if there does not exist a triple $x,y,z$ such that, for any two of them, there exists a path that avoids the closed neighbourhood of the third one. The interval graphs are exactly the chordal {\em AT-free} graphs~\cite{LeB62}. The complexity of (exact and approximate) diameter computation within AT-free graphs was studied in~\cite{CDDH+01} (see also~\cite{Dra99,HeK02}), where the authors emphasize a kind of duality between AT-free graphs and chordal graphs. For instance, on both graph classes, two consecutive executions of LexBFS always yield a vertex whose eccentricity is within one of the diameter -- this is the so-called 2-sweep LexBFS algorithm, see Fig.~\ref{fig:2-sweep}. However, there is no constant $c$ such that $c$ consecutive executions of LexBFS on these graph classes always output the exact diameter. The authors from~\cite{CDDH+01} further sketch a reduction from OV to diameter computation within AT-free graphs and chordal graphs, as evidence that the diameter problem on these graph classes cannot be solved in linear time. This same reduction was revisited in~\cite{BCH16} in order to prove that, indeed, under SETH there is no truly subquadratic algorithm for computing the diameter on chordal graphs. But the same {\em cannot} be done for AT-free graphs, because for the latter, the reduction in~\cite{CDDH+01} from OV to diameter computation already runs in $\Omega(n^2)$ time. 

\begin{itemize}
	\item {\bf Our first main result} is an ${\cal O}(m^{3/2})$-time algorithm in order to compute all the eccentricities (and so, the diameter) in an AT-free graph (Theorem~\ref{thm:main}).
\end{itemize}

The proof of this above result stays simple. Nevertheless, it comes to us as surprise given the evidence for SETH-hardness in~\cite{CDDH+01}. On dense graphs -- with $m \geq c \cdot n^2$ edges, for some constant $c$ -- our algorithm does no better than the classic ${\cal O}(nm)$-time algorithm for All-Pairs Shortest-Paths. We suspect this to be unavoidable, or to be more precise, that the diameter problem on dense AT-free graphs is computationally equivalent with Boolean matrix multiplication. However, we were unable to prove this, and we leave it as an intriguing open question. For claw-free AT-free graphs, there is a linear-time algorithm in order to compute the diameter and all the vertices of minimum eccentricity~\cite{HeK02}.

\smallskip
The {\em ball hypergraph} of $G$ has for hyperedges the balls of all possible centers and radii in $G$.
On our way, we observe that for the family of ball hypergraphs of AT-free graphs, classic geometric parameters such as the Helly number and the VC-dimension are unbounded (Proposition~\ref{prop:unbounded}).
It sets AT-free graphs apart from most known graph classes with a truly subquadratic algorithm for the diameter problem~\cite{DHV20,DuD19+}. We initiate the complexity of diameter computation within graph classes sharing a common property with AT-free graphs. In this paper, we only study such classes which are subclasses of chordal graphs. This is for the following two main reasons. On one hand, under SETH, there is no truly subquadratic algorithm for computing the diameter, already for chordal graphs. On the other hand, chordal graphs are more structured than general graphs, thereby making easier the design and the analysis of our algorithms. We summarise our results for subclasses of chordal graphs (we postpone their technical definitions to appropriate places throughout the paper):
\begin{itemize}
\item For every chordal graph with {\em asteroidal number} at most $k$, there is a randomized ${\cal O}(km\log^2{n})$-time algorithm in order to compute the diameter with high probability (Theorem~\ref{thm:asteroidal-num}). In contrast to this positive result, it is easy to prove that under SETH, there is no truly subquadratic algorithm for diameter computation within $k$-AT-free chordal graphs, for every $k \geq 2$ (Proposition~\ref{prop:k-at});
\item For every chordal {\em dominating pair} graph, there is a truly subquadratic algorithm for computing the diameter (Theorem~\ref{thm:dom-pair-chordal}). For the larger class of chordal graphs with a dominating shortest path, there is a linear-time algorithm for computing a diametral pair if the diameter is at least four; however, already for split graphs with a dominating edge, under SETH there is no truly subquadratic algorithm for computing the diameter (Theorem~\ref{thm:dicho});
\item For every chordal graph with a dominating triple, if the diameter is at least $10$, then there is a linear-time algorithm in order to compute a diametral pair (Theorem~\ref{thm:dom-triple}).
\end{itemize}
We stress that all the aforementioned graph classes generalize interval graphs in various ways, that are incomparable one with another. For chordal Helly graphs and chordal graphs of bounded VC-dimension: two other generalizations of interval graphs that are uncomparable with each other and with the other subclasses presented above, there also exist truly subquadratic algorithms for the diameter problem~\cite{DuD19+}. We left open whether Theorem~\ref{thm:asteroidal-num} can be derandomized. It is also open whether the lower bound of $10$ on the diameter for Theorem~\ref{thm:dom-triple} is tight. Finally, we insist on the dichotomy of Theorem~\ref{thm:dicho}, where we show that the {\em only} difficulty for a chordal graph with a dominating shortest-path is to decide whether the diameter is either two or three. By contrast, for any $d \geq 1$, under SETH there is no truly subquadratic algorithm for distinguishing the chordal graphs of diameter $\leq 2d$ from those of diameter $\geq 2d+1$ (to see that, just start from any split graph, and replace every vertex of its stable set by a path of length $d-1$).
%
%\medskip
%\noindent
%Some of our results and techniques still hold beyond chordal graphs.
%In Sec.~\ref{sec:gal-graphs}, we extend Theorems~\ref{thm:dicho} and~\ref{thm:dom-triple} to all graphs with a dominating shortest-path and to graphs of asteroidal number at most three, respectively (Theorems~\ref{thm:spd} and~\ref{thm:dt}).

\section{All eccentricities for AT-free graphs}\label{sec:algo}

\subsection{Preliminaries}\label{sec:prelim}

We start recalling a few results from prior work, that we will use in this paper.

\paragraph{Graph searches.}The Lexicographic Breadth-First Search (LexBFS) is a standard algorithmic procedure, that runs in linear time~\cite{RTL76}. We give a pseudo-code in Fig.~\ref{fig:lexbfs}. For a given graph $G = (V,E)$ and a start vertex $u$, we denote $LexBFS(u)$ the corresponding execution of LexBFS. Its output is a numbering $\sigma$ over the vertex-set (namely, the reverse of the ordering in which vertices are visited during the search). In particular, if $\sigma(i) = x$, then $\sigma^{-1}(x) = i$. 

\begin{figure}[h!]\centering
\includegraphics[width=.7\textwidth]{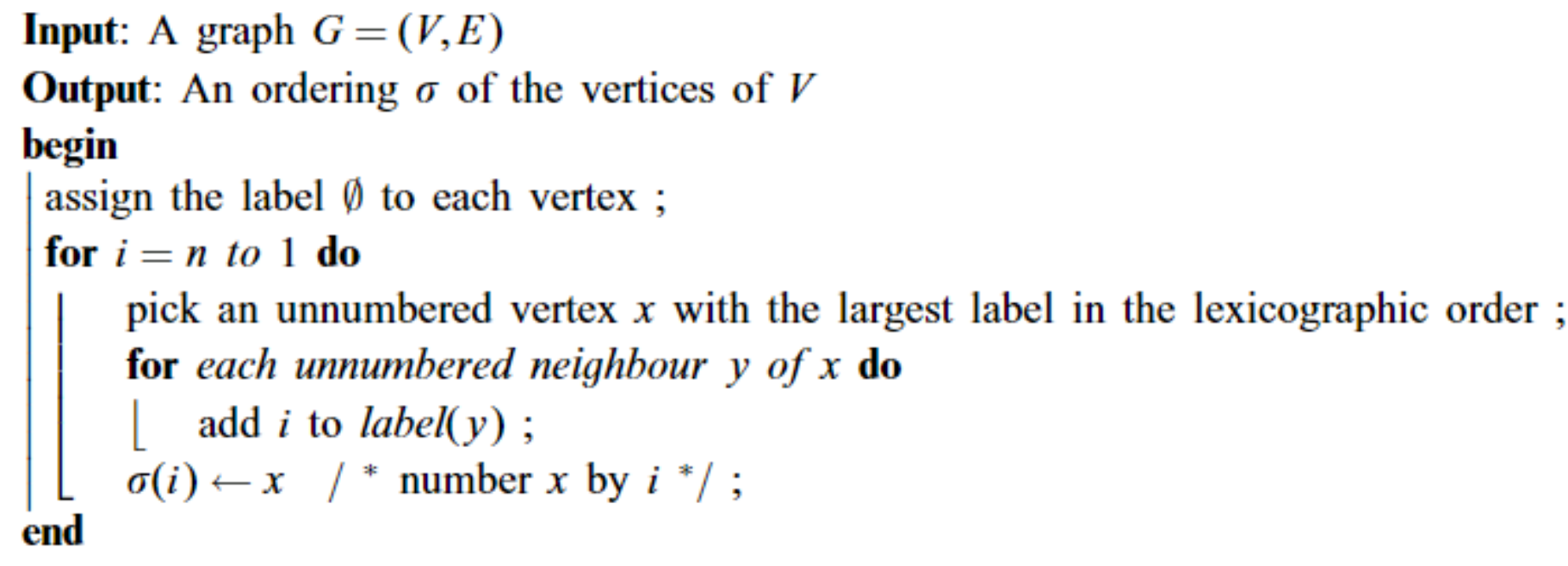}
\caption{Algorithm LexBFS~\cite{RTL76}.}
\label{fig:lexbfs}
\end{figure}

\smallskip
\noindent
The $2$-sweep LexBFS consists in two consecutive execution of LexBFS, with the start vertex of the second execution being the last one visited during the first LexBFS. See also Fig.~\ref{fig:2-sweep}. More generally, for any positive integer $c$, the algorithm $c$-sweep LexBFS consists in $c$ consecutive applications of LexBFS, with for any $i \geq 2$, the start vertex of the $i^{th}$ execution being the last vertex visited during the $(i-1)^{th}$. We prove most of our results using $3$-sweep LexBFS, but often use the known properties of $2$-sweep LexBFS in our proofs.

\begin{figure}[h!]\centering
\includegraphics[width=.4\textwidth]{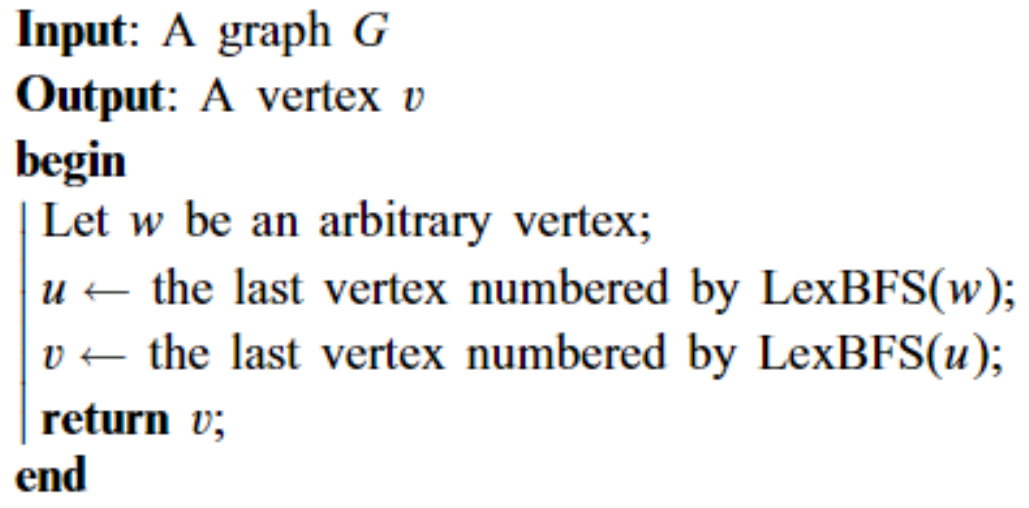}
\caption{Algorithm $2$-sweep~\cite{CDDH+01}.}
\label{fig:2-sweep}
\end{figure}

\paragraph{AT-free graphs.}
We now recall several useful properties of LexBFS orderings within AT-free graphs. 
A {\em dominating pair} in a graph is a pair $(u,v)$ of two vertices such that, for every vertex $x$, every $uv$-path intersects its closed neighbourhood $N[x]$. 

\begin{lemma}[\cite{COS99}]\label{lem:dominating}
For an AT-free graph $G = (V,E)$, let $u$ be the last vertex visited during a LexBFS, and let $\sigma = LexBFS(u)$. Then, for every vertex $y$, the pair $(u,y)$ is dominating for the subgraph induced by $\{ z \in V \mid \sigma^{-1}(z) \geq \sigma^{-1}(y) \}$. In particular, if $\sigma(1) = v$, then $(u,v)$ is a dominating pair for $G$.
\end{lemma}

\begin{lemma}[\cite{DKL17}]\label{lem:comp}
For an AT-free graph $G = (V,E)$, let $u$ be the last vertex visited during a LexBFS, and let $\sigma = LexBFS(u)$.
If $x,y \in V$ are such that: $xy \notin E, \ dist(u,x) = dist(u,y) = i, \ \text{and} \ \sigma^{-1}(x) < \sigma^{-1}(y)$, then we have $N(x) \cap \{ z \in V \mid dist(u,z) = i-1 \} \subseteq N(y) \cap \{ z \in V \mid dist(u,z) = i-1 \}$. In particular, $dist(x,y) \leq 2$.
\end{lemma}

%This above Lemma~\ref{lem:dominating} implies that for any AT-free graph, the last vertex visited during a $2$-sweep LexBFS has eccentricity within one of the diameter of $G$~\cite{CDDH+01}. Further results on the structure of AT-free graphs are known, namely:
%
%\begin{lemma}[\cite{CDDH+01}]\label{lem:diam-pair}
%Let $G = (V,E)$ be an AT-free graph with $diam(G) = k > 2$. If $e(v) = k-1$, where $v$ is the vertex returned by $2$-sweep LexBFS, and $u',v'$ achieve the diameter where $dist(u,u') \leq dist(u,v')$ then:
%\begin{itemize}
%\item $dist(u,v) = dist(u,v') = dist(u',v) = k-1$;
%\item $uu',vv' \in E$.
%\end{itemize}
%\end{lemma}

\subsection{The algorithm}

In what follows, we prove that all eccentricities in an AT-free graphs can be computed in truly subquadratic time (Theorem~\ref{thm:main}). Before that, we need to prove a few intermediate lemmas.

\begin{lemma}\label{lem:furthest-vertex}
Let $(u,v)$ be a dominating pair in a graph $G=(V,E)$ (not necessarily AT-free).
Every vertex $x \in V$ s.t. $e(x) \geq 3$ is at distance $e(x)$ from a vertex in $N[u] \cup N[v]$.
\end{lemma}

\begin{proof}
Let $x \in V$ be s.t. $e(x) \geq 3$, and let $y \in V$ s.t. $dist(x,y) = e(x)$.
In what follows, we assume $y \notin N[u] \cup N[v]$ (otherwise, we are done).
Name $P$ a shortest $uv$-path in $G$ s.t. $dist(y,P)$ is minimized.
In particular, we have $y \in V(P)$ if and only if $dist(u,v) = dist(u,y) + dist(y,v)$.
We pick any two vertices $x^* \in N[x] \cap V(P)$ and $y^* \in N[y] \cap V(P)$, that always exists because, by the hypothesis, $P$ is dominating.
By symmetry, we may assume that $y^*$ is on the subpath of $P$ between $u$ and $x^*$. 
We claim that $e(x) = dist(u,x)$, that will prove the lemma.
Suppose for the sake of contradiction that it is not the case.
Consider the (not necessarily simple) $uv$-path $P'$ that is obtained from the concatenation of $P[v,x^*]$: the subpath of $P$ between $v$ and $x^*$, with an arbitrary shortest $xu$-path $Q$.
Since $(u,v)$ is a dominating pair, $N[y] \cap V(P') \neq \emptyset$.
However, we prove in what follows that $y$ has no neighbour on $P[v,x^*]$. 
Suppose by contradiction that there exists a vertex $z \in N[y]$ on the subpath $P[v,x^*]$. Then, $dist(y^*,z) \leq 2$.
Observe that we have $dist(y^*,z) = dist(y^*,x^*) + dist(x^*,z)$.
Furthermore, $x^* \notin N[y]$ because we assume $dist(x,y) = e(x) \geq 3$.  Since $P$ is a shortest $uv$-path, it follows that $y^*,z \in N(y) \cap N(x^*)$ and $dist(y^*,z) = 2$.
But then, we could replace $x^*$ by $y$ on our shortest $uv$-path, thus contradicting the minimality of $dist(y,P)$. As a result, and since $N[y] \cap V(P') \neq \emptyset$, we must have $N[y] \cap V(Q) \neq \emptyset$. Since $Q$ is a shortest $ux$-path and we also have $u \notin N[y]$, we derive from the above that $dist(x,y) \leq dist(x,u)$, thus contradicting our assumption that $e(x) > dist(x,u)$. 
\end{proof}

By using this above Lemma~\ref{lem:furthest-vertex}, we will prove that in order to compute all eccentricities, it suffices to perform a BFS from all vertices in the first and two last distance layers of some shortest-path tree, whose root is the vertex outputted by the algorithm $3$-sweep LexBFS.
However, there does not seem to be a simple way in order to upper bound the number of vertices in these layers.
We complete our result, as follows:

\begin{lemma}\label{lem:dom-ecc}
For an AT-free graph $G = (V,E)$, let $u$ be the last vertex visited during a LexBFS, and let $\sigma = LexBFS(u)$.
Let $x,y \in V$ be s.t.: $xy \notin E, dist(u,x) = dist(u,y) = i, \ \text{and} \ \sigma^{-1}(x) < \sigma^{-1}(y)$. For any vertex $z \in V$, if $dist(u,z) < i$ then we always have $dist(y,z) \leq dist(x,z)$. 

In particular, if $G$ has no universal vertex and $dist(u,x) = dist(u,y) = e(u)$, then $e(x) \geq e(y)$.
\end{lemma}

\begin{proof}
Let $z \in V$ be such that $dist(u,z) < i$. In particular, $\sigma^{-1}(z) > \sigma^{-1}(y) > \sigma^{-1}(x)$. For any shortest $xz$-path $P$, we claim that $N[y] \cap V(P) \neq \emptyset$. Note that, since we assume $xy \notin E$, this will prove that $dist(y,z) \leq dist(x,z)$. In order to prove the claim, let $x' \in V(P)$ minimize $\sigma^{-1}(x')$. By construction, $\sigma^{-1}(x') \leq \sigma^{-1}(x) < \sigma^{-1}(y)$. Let $P_{x'}$ be the subpath of $P$ between $x'$ and $z$. We complete $P_{x'}$ into an $ux'$-path $P'$ (not necessarily simple) by adding to it a shortest $uz$-path $Q$. In particular, all the vertices $w \in V(P') = V(P_{x'}) \cup V(Q)$ satisfy $\sigma^{-1}(w) \geq \sigma^{-1}(x')$. Since we also have $\sigma^{-1}(y) > \sigma^{-1}(x')$, it follows by Lemma~\ref{lem:dominating}, $N[y] \cap V(P') \neq \emptyset$. Furthermore, since any vertex $w \in V(Q) \setminus \{z\}$ satisfies $dist(u,w) \leq dist(u,z) - 1 \leq dist(u,y) - 2$, we get $N[y] \cap (V(Q) \setminus \{z\}) = \emptyset$. The latter proves, as claimed, $N[y] \cap V(P) \supseteq N[y] \cap V(P_{x'}) \neq \emptyset$.

Finally, let us assume that $dist(u,x) = dist(u,y) = e(u)$. If $z \in V$ is such that $dist(u,z) < e(u)$, then we proved above that we have $dist(y,z) \leq  dist(x,z)$. Otherwise, $dist(u,z) = e(u)$, and by Lemma~\ref{lem:comp} we get that $dist(y,z) \leq 2$. Altogether combined, $e(y) \leq \max\{e(x),2\}$.
If moreover, $G$ has no universal vertex, then $e(x) \geq 2$, and so, we conclude that $e(y) \leq e(x)$. 
\end{proof}

If we execute three consecutive LexBFS instead of two, then a weaker converse of Lemma~\ref{lem:dom-ecc} for {\em lower} distance layers (instead of upper ones) also holds.
Namely:

\begin{lemma}\label{lem:dom-ecc-rev}
For an AT-free graph $G = (V,E)$, let $u$ be the last vertex visited during a LexBFS, let $\sigma = LexBFS(u)$, and let $v = \sigma(1)$ be the last vertex visited during $LexBFS(u)$.
Let $x,y \in V$ be s.t.: $xy \notin E, dist(v,x) = dist(v,y) = i, \ \text{and} \ \sigma^{-1}(x) < \sigma^{-1}(y)$. For any vertex $z \in V$, if $dist(v,z) > i$ then we always have $dist(y,z) \leq \max\{dist(x,z),2\}$. 
\end{lemma}

\begin{proof}
Since we have $\sigma^{-1}(x) < \sigma^{-1}(y)$, $dist(u,x) \geq dist(u,y) = j$. Now, let $z \in V$ be s.t. $dist(v,z) > i$. If furthermore, $dist(u,z) < j$, then we claim that $dist(y,z) \leq dist(x,z)$. Indeed, if $dist(u,x) = dist(u,y)$, then the claim follows from Lemma~\ref{lem:dom-ecc}. From now on, we assume $dist(u,x) > dist(u,y)$. Any shortest $xz$-path must contain a vertex $x_z$ s.t. $dist(u,x_z) = dist(u,y)$. We may further assume $x_z \in N(x)$ (otherwise, since by Lemma~\ref{lem:comp} we have $dist(y,x_z) \leq 2$, $dist(y,z) \leq 2 + dist(x_z,z) \leq dist(x,z)$, and we are done). In particular, $dist(u,x) = j+1$. In the same way, we may assume that $y,x_z$ are nonadjacent, and that any neighbour $x' \in N(x_z)$ on a shortest $x_zz$-path satisfies $dist(u,x') = j-1$. Then, consider an $xu$-path $P$ starting with $[x,x_z,x']$ and continuing with any shortest $x'u$-path. By Lemma~\ref{lem:dominating}, $P$ dominates all vertices $w$ s.t. $\sigma^{-1}(w) \geq \sigma^{-1}(x)$. In particular, $N[y] \cap V(P) \neq \emptyset$. Since we assume $x,x_z \notin N[y]$, it implies $x' \in N(y)$. In such case, $dist(y,z) \leq 1 + dist(x',z) < dist(x,z)$, thereby proving our claim.

We are left handling with the case $dist(u,z) \geq j$. In order to prove the lemma, it suffices to prove that $dist(y,z) \leq 2$. If $dist(u,z) = j$, then this follows from Lemma~\ref{lem:comp}. Thus, from now on we assume $dist(u,z) > j$. In order to conclude here, we need the following observation: for every vertex $w$, we always have $dist(u,v) \leq dist(u,w) + dist(w,v) \leq dist(u,v) + 2$. Indeed, the first inequality follows from the triangular inequality. The second inequality follows from Lemma~\ref{lem:dominating} saying that any shortest $uv$-path is dominating (to see that, consider any $w^* \in N[w]$ on such shortest-path, and observe that $dist(u,w) + dist(w,v) \leq dist(u,w^*) + dist(w^*,v) + 2 = dist(u,v) + 2$). In our case, $dist(u,z) + dist(z,v) \geq j+1 + i+1 = dist(u,y) + dist(y,v) + 2$. Therefore, $y$ is on a shortest $uv$-path $Q$, and in addition $dist(u,z) = j+1, \ dist(v,z) = i+1$. By Lemma~\ref{lem:dominating}, there exists a $z^* \in N[z] \cap V(Q)$. The only possibility w.r.t. $dist(u,z), \ dist(v,z)$ is to have $z^* = y$, and so, $dist(y,z) = 1$.
\end{proof}

We are now ready to prove the main result in this section:

\begin{theorem}\label{thm:main}
For every $m$-edge AT-free graph $G=(V,E)$, we can compute the eccentricities in ${\cal O}(m^{\frac 3 2})$ time.
\end{theorem}

\begin{proof}
Let $u$ be the last vertex visited during a LexBFS, and let $\sigma = LexBFS(u)$.
Similarly, let $v = \sigma(1)$ be the last vertex visited during $LexBFS(u)$, and let $\tau = LexBFS(v)$. Set $d = e(v)$. For any $0 \leq i \leq d$, we define $L_i = \{ w \in V \mid dist(v,w) = i \}$.
We compute the set $A$ from $L_1 \cup \{v\}$ and $\sigma$ as follows.
We scan all the vertices $v' \in L_1 \cup \{v\}$ by increasing value of $\sigma^{-1}(v')$, removing from this set the non-neighbours of $v'$. 
Similarly, we compute the set $B$ (resp., $C$) from $L_{d-1}$ (resp., $L_d$) and $\tau$ as follows.
We scan all the vertices $u' \in L_{d-1}$ (resp., $u' \in L_{d}$) by increasing value of $\tau^{-1}(u')$, removing from this set the non-neighbours of $u'$.
The three of $A,B,C$ can be computed in linear time, and since they are cliques, their cardinality is in ${\cal O}(\sqrt{m})$.  

For every vertex $w \in V$, we claim that we have:
$$e(w) = \begin{cases}
1 \ \text{if} \ w \ \text{is a universal vertex} \\
\max \{2\} \cup \{ dist(w,x) \mid x \in A \cup B \cup C \} \ \text{otherwise.}
\end{cases}$$
Indeed, $w$ is universal if and only if $e(w) = 1$. The above formula is also trivially true if $e(w) = 2$. Therefore, let us assume in what follows $e(w) \geq 3$. By Lemma~\ref{lem:dominating} $(\tau(1),v)$ is a dominating pair, and therefore by Lemma~\ref{lem:furthest-vertex}, $w$ is at distance $e(w)$ from some vertex in $N[\tau(1)] \cup N[v]$. Note that $N[\tau(1)] \cup N[v] \subseteq \{v\} \cup L_1 \cup L_{d-1} \cup L_d$. Let $y \in \{v\} \cup L_1 \cup L_{d-1} \cup L_d$ s.t. $e(w) = dist(y,w)$. We consider in what follows four different cases:
\begin{itemize}
\item Case $y = v$. By construction, $v \in A$. Therefore, the above formula for $e(w)$ holds in this case.
\item Case $y \in L_1$. We assume w.l.o.g. $\sigma^{-1}(y)$ is minimized for this property. Suppose for the sake of contradiction $y \notin A$ (otherwise, we are done). By the construction of $A$, there exists a $x \in L_1$ s.t. $\sigma^{-1}(x) < \sigma^{-1}(y)$ and $x,y$ are nonadjacent. Furthermore, $dist(v,w) > 1$ because otherwise, we would get $e(w) = dist(y,w) \leq 2$. By Lemma~\ref{lem:dom-ecc-rev}, $dist(y,w) \leq \max\{2,dist(x,w)\}$. Since in addition we have $dist(y,w) = e(w) \geq 3$, we conclude in this case that we have $dist(x,w) = e(w)$, that contradicts the minimality of $\sigma^{-1}(y)$.
\item Case $y \in L_d$. The proof is quite similar as for the previous case, but slightly simpler. We assume w.l.o.g. $\tau^{-1}(y)$ is minimized for this property. Suppose for the sake of contradiction $y \notin C$. By the construction of $C$, there exists a $x \in L_d$ s.t. $\tau^{-1}(x) < \tau^{-1}(y)$ and $x,y$ are nonadjacent. Furthermore, $dist(v,w) < d$ because otherwise, we would get by Lemma~\ref{lem:comp} $e(w) = dist(y,w) \leq 2$. By Lemma~\ref{lem:dom-ecc-rev}, $dist(y,w) \leq dist(x,w)$. Since in addition we have $dist(y,w) = e(w)$, we conclude in this case that we have $dist(x,w) = e(w)$, that contradicts the minimality of $\tau^{-1}(y)$.
\item Case $y \in L_{d-1}$. Here also, the proof is essentially the same as for the previous case. In fact, we only need to detail the special case $w \in L_d$. We claim that we have $e(w) = 3$. Indeed, let $x \in N(w) \cap L_{d-1}$. By Lemma~\ref{lem:comp} we have $dist(w,y) \leq 1 + dist(x,y) \leq 3$. We may further assume w.l.o.g. $d-1 > 1$ (otherwise, $y \in \{v\} \cup L_1$). But then, $dist(v,w) = d \geq 3 = e(w)$.
\end{itemize}
Finally, we claim that the above formula can be computed, for all vertices, in total ${\cal O}(m^{3/2})$ time. Indeed, it suffices to perform a BFS from every vertex of $A \cup B \cup C$, and we observed above that there are only ${\cal O}(\sqrt{m})$ many such vertices.
\end{proof}

\subsection{Digression: properties of the neighbourhood hypergraph}

For a graph $G=(V,E)$, its neighbourhood hypergraph is ${\cal N}(G) = (V, \{ N[v] \mid v \in V\})$. Note that it is a subhypergraph of the ball hypergraph (as defined in the introduction). The {\em Helly number} of ${\cal N}(G)$ is the smallest $k$ s.t. every family of $k$-wise intersecting neighbourhoods of $G$ ({\it i.e.}, hyperedges) have a nonempty common intersection. Its {\em VC-dimension} is the largest $d$ s.t., for some vertex-subset $X$ of cardinality $d$, for every subset $Y \subseteq X$, there exists a $v \in V$ s.t. $N[v] \cap X = Y$ (we say that $X$ is shattered by ${\cal N}(G)$). It was proved recently that a bounded Helly number or bounded VC-dimension for the ball hypergraph (and so, for the neighbourhood hypergraph as well) implies fast diameter and radius computation algorithms~\cite{DHV20,DuD19+}. Actually, for most graph classes for which we know how to compute the diameter in truly subquadratic time, one of these two parameters above is always bounded. For instance, on interval graphs, both parameters are at most two. We prove that such property does not hold for the AT-free graphs:

\begin{proposition}\label{prop:unbounded}
There are AT-graphs whose neighbourhood hypergraph has a Helly number (resp., VC-dimension) that is arbitrarily large.
\end{proposition}
\begin{proof}
We take the opportunity to recall the reduction in~\cite{CDDH+01} from OV to AT-free graphs. Let $A,B$ be two families of sets over a universe $C$. The graph $H_{A,B,C}$ has vertex-set $A \cup B \cup C$. Its edge-set is as follows. The sets $A,B,C$ are cliques. For every $a \in A$ and $c \in C$, $a$ and $c$ are adjacent if and only if $c \in a$. In the same way, for every $b \in B$ and $c \in C$, $b$ and $c$ are adjacent if and only if $c \in b$. In~\cite{CDDH+01}, the authors observed that $H_{A,B,C}$ is a cocomparability graph (and so, it is AT-free). We prove that for some suitable $A,B$ and $C$, the neighbourhood hypergraph of $H_{A,B,C}$ has arbitrarily large Helly number (resp., VC-dimension). 

First, for any fixed $d$, we consider a family $A$ over some universe $C$ and VC-dimension $\geq d$. Let $X$ be of cardinality $d$ and shattered by $A$. For every $Y \subseteq X$, there exists a set $a \in A$ s.t. $a \cap X = Y$. In particular, identifying $X \subseteq C$ with a vertex-subset of $H_{A,\emptyset,C}$, we also get $N[a] \cap X = Y$. As a result, the VC-dimension of ${\cal N}(H_{A,\emptyset,C})$ is at least $d$ (since $X$ is shattered by this hypergraph). 

In the same way, for any fixed $k$, let $F$ be a minimal family of Helly number $\geq k+1$ over some universe $C$. By minimality of the family, all sets in $F$ $k$-wise intersect, but they have an empty common intersection. We arbitrarily bipartition $F$ into nonempty subfamilies $A$ and $B$. Then, we claim that the neighbourhood hypergraph of $H_{A,B,C}$ has Helly number at least $k+1$. Indeed, by construction the neighbour sets $N[a], a \in A$ and $N[b], b \in B$ $k$-wise intersect. However, since their common intersection must be in $C$, then by the choice of $F$ the latter is empty.
\end{proof}

\section{Chordal graphs with bounded asteroidal number}

Let us recall that an asteroidal set in a graph $G=(V,E)$ is an independent set $A \subseteq V$ s.t., for every $v \in A$, all vertices in $A \setminus \{v\}$ are in the same connected component of $G \setminus N[v]$. In particular, an AT is exactly an asteroidal set of cardinality three. The {\em asteroidal number} of a graph $G$ is the largest cardinality of its asteroidal sets. Since with this terminology, the interval graphs are exactly the chordal graphs of asteroidal number two, the following result generalizes the celebrated linear-time algorithm for computing the diameter in this graph class:

\begin{theorem}\label{thm:asteroidal-num}
There is a randomized ${\cal O}(km\log^2{n})$-time algorithm for computing w.h.p. the diameter of chordal graphs with asteroidal number at most $k$.
\end{theorem}

Our results are based on a general framework for computing the diameter of chordal graphs~\cite{DuD19+}.
We recall that a split graph $G$ is a graph whose vertex-set can be bipartitioned into a clique $K$ and a stable set $S$. We may assume $G$ to be given under its {\em sparse representation}, defined in~\cite{DHV19} as being the hypergraph $(K \cup S,\{N_G[s] \mid s \in S\})$. The {\sc Split-OV} problem is a special case of OV where $A = \{ N_G[a] \mid a \in S_A \}$ and $B = \{ N_G[b] \mid b \in S_B \}$, for some split graph $G$ and for some partition $S_A,S_B$ of its stable set. An instance of {\sc Split-OV} can be encoded as a triple $(G,S_A,S_B)$, with $G$ being given under its sparse representation. Note that the size of such instance is dominated by $\ell = \sum_{s \in S_A \cup S_B} |N_G[s]|$, that is $\leq 2m + n$. 

\begin{theorem}[Theorem $8'$ from~\cite{DuD19+}]\label{thm:framework}
For a subclass ${\cal C}$ of chordal graphs, let ${\cal S}$ be the subclass of all split graphs that are induced subgraphs of a chordal graph in ${\cal C}$. If for every $(G,S_A,S_B)$, with $G \in {\cal S}$ connected, we can solve {\sc Split-OV} in ${\cal O}(\ell^b)$ time, for some $b \geq 1$, then there is a randomized ${\cal O}(m^b\log^2{n})$-time algorithm for computing w.h.p. the diameter of chordal graphs in ${\cal C}$.
\end{theorem}

Recall that having asteroidal number at most $k$ is a hereditary property.
Hence, by the above Theorem~\ref{thm:framework}, in order to prove Theorem~\ref{thm:asteroidal-num}, it suffices to solve {\sc Split-OV} in ${\cal O}(k\ell)$ time for connected split graphs of asteroidal number at most $k$. Before presenting such an algorithm (Proposition~\ref{prop:slit-an}), we need a few preparatory lemmas. 

\begin{lemma}\label{lem:inclusion}
For a split graph $G$ with stable set $S$, and $A \subseteq S$ of cardinality $|A| \geq 3$, $A$ is an asteroidal set if and only if the neighbour sets $N(a), \ a \in A$ are pairwise incomparable w.r.t. inclusion. 
\end{lemma}

\begin{proof}
On one direction, if $N(a') \subseteq N(a)$ for some distinct $a,a' \in A$, then we claim that $A$ is not an asteroidal set. Indeed, $a'$ is isolated in $G \setminus N[a]$, and so disconnected from $A \setminus \{a,a'\}$, that is nonempty because $|A| \geq 3$. On the other direction, let us assume the neighbour sets $N(a), \ a \in A$ to be pairwise incomparable w.r.t. inclusion. Let $a \in A$ be arbitrary. By the hypothesis all the subsets $N(a') \setminus N(a), a' \in A \setminus \{a\}$ are nonempty. Since $G$ is a split graph, $N(A \setminus \{a\}) \setminus N(a) = \bigcup_{a' \in A \setminus \{a\}} N(a') \setminus N(a)$ is a clique. Therefore, the vertices of $A \setminus \{a\}$ are in the same connected component of $G \setminus N[a]$. 
\end{proof}

\begin{lemma}\label{lem:compute-min}
Let ${\cal F}$ be a family of pairwise different subsets, and $\ell = \sum_{S \in {\cal F}}|S|$.
If there are at most $k$ inclusionwise minimal subsets in ${\cal F}$, then all these subsets can be computed in total ${\cal O}(k\ell)$ time.
\end{lemma}

\begin{proof}
It suffices to prove that in ${\cal O}(\ell)$ time, we can compute an inclusionwise minimal subset and remove all its supersets from ${\cal F}$. Specifically, let $S \in {\cal F}$ be of minimum cardinality. Since all sets in ${\cal F}$ are pairwise different, $S$ is inclusion minimal. Furthermore, if we first mark all elements in $S$, and we store the cardinality $|S|$ of this set, then we can remove all its supersets from ${\cal F}$ by scanning the family in total ${\cal O}(\ell)$ time.  
\end{proof}

Combining Lemmas~\ref{lem:inclusion} and~\ref{lem:compute-min}, we get the following algorithm for {\sc Split-OV} on split graphs of asteroidal number at most $k$:

\begin{proposition}\label{prop:slit-an}
For every $(G,S_A,S_B)$ where $G$ has asteroidal number at most $k \geq 2$, we can solve {\sc Split-OV} in ${\cal O}(k\ell)$ time.
\end{proposition}

\begin{proof}
Two vertices $a,a' \in S_A$ are {\em twins} if $N(a) = N(a')$.
It is now folklore that by using partition refinement techniques, we can compute all twin classes of $S_A$ in ${\cal O}(\ell)$ time~\cite{HMPV00}. Note also that in order to solve {\sc Split-OV}, it is sufficient to keep only one vertex per twin class of $S_A$. Thus, from now on we assume that the neighbour sets $N(a), a \in S_A$ are pairwise different. Observe that if $a,a' \in S_A$ and $b \in S_B$ satisfy $N(a') \subset N(a)$ and $N(a) \cap N(b) = \emptyset$, then $N(a') \cap N(b) = \emptyset$. In particular, we may further restrict $S_A$ to the vertices $a$ s.t. $N(a)$ is inclusion wise minimal. By Lemma~\ref{lem:inclusion}, there are at most $k$ such vertices. Therefore, by Lemma~\ref{lem:compute-min}, we can compute all these vertices in total ${\cal O}(k\ell)$ time. Finally, for each $a \in S_A$ s.t. $N(a)$ is inclusion wise minimal, we can compute a vertex $b \in S_B$ s.t. $N(a) \cap N(b) = \emptyset$ (if any) simply by scanning $N(a)$ and all the neighbour sets $N(b), \ b \in S_B$. It takes total ${\cal O}(\ell)$ time.
\end{proof}

Theorem~\ref{thm:asteroidal-num} now follows from Theorem~\ref{thm:framework} and Proposition~\ref{prop:slit-an}.\qed

\smallskip
\noindent
A $k$-AT in a graph is a triple $x,y,z$ of vertices such that, for any two of them, there exists a path that avoids the {\em ball of radius $k$} of the third one, where the latter is defined, for any vertex $v$, as $N^k[v] = \{ u \mid dist(u,v) \leq k \}$. A graph is $k$-AT-free if it does not contain any $k$-AT. Note that in particular, the AT-free graphs are exactly the $1$-AT-free graphs. This generalization of AT-free graphs was first proposed in~\cite{MaF16}, where the performances of BFS for the latter were studied. Before concluding this section, we make the following simple observation:

\begin{proposition}\label{prop:k-at}
Under SETH, there is no truly subquadratic algorithm for computing the diameter of $k$-AT-free chordal graphs, for every $k \geq 2$.
\end{proposition}

\begin{proof}
Under SETH, there is no truly subquadratic algorithm for computing the diameter of split graphs~\cite{BCH16}.
Therefore, it suffices to prove that split graphs are $k$-AT-free, for every $k \geq 2$.
Indeed, this is because, for a split graph with clique $K$, the ball of radius $2$ for {\em any} vertex contains $K$. In particular, its removal either leaves an empty graph, a singleton, or an independent set. As a result, there can be no $k$-AT in a split graph, for every $k \geq 2$.
\end{proof}

\section{Chordal graphs with a dominating shortest-path}\label{sec:dom-sp}

Recall (see Lemma~\ref{lem:dominating}) that every AT-free graph has a dominating pair. A {\em dominating pair} graph is one s.t. every induced subgraph has a dominating pair. We study diameter computation within chordal dominating pair graphs in Sec.~\ref{sec:dompair}. A weaker property for a graph is to have a dominating shortest-path. We first study chordal graphs with a dominating shortest-path, in Sec.~\ref{sec:duality}, for which we derive an interesting dichotomy result.

\subsection{Dichotomy theorem}\label{sec:duality}

The purpose of this section is to prove the following result:

\begin{theorem}\label{thm:dicho}
For every chordal graph with a dominating shortest-path, if the algorithm $3$-sweep LexBFS outputs a vertex of eccentricity $d \geq 3$, then it is the diameter. In particular, there is a linear-time algorithm for deciding whether the diameter is at least four on this graph class.

However, already for split graphs with a dominating edge, under SETH there is no truly subquadratic algorithm for deciding whether the diameter is either two or three.
\end{theorem}

We start with the following easy lemma:

\begin{lemma}[~\cite{DeK95}]\label{lem:dom-dp}
If a graph has a dominating shortest-path, then it has a dominating diametral path.
\end{lemma}

Roughly, the first part of Theorem~\ref{thm:dicho} follows from this above Lemma~\ref{lem:dom-dp} combined with the following properties of LexBFS on chordal graphs:

\begin{lemma}[~\cite{CDDH+01}]\label{lem:lexbfs-chordal}
Let $u$ be the vertex of a chordal graph $G$ last visited by a LexBFS, and let $x,y$ is a pair of vertices such that $dist(x,y) = diam(G)$. If $e(u) < diam(G)$ then $e(u)$ is even, $dist(u,x) = dist(u,y) = e(u)$ and $e(u) = diam(G)-1$.
\end{lemma}

\begin{corollary}[~\cite{CDDH+01}]\label{cor:lexbfs-chordal}
If the vertex $u$ of a chordal graph $G$ last visited by a LexBFS has odd eccentricity, then $e(u) = diam(G)$.
\end{corollary}

\begin{proof}[Proof of Theorem~\ref{thm:dicho}]
Let $u$ be last visited by a LexBFS, and let $v$ be last visited by $LexBFS(u)$.
Assume toward a contradiction $e(v) < diam(G)$.
Since $e(u) = dist(u,v) \leq e(v)$, and by Lemma~\ref{lem:lexbfs-chordal}, $e(u) \geq diam(G)-1$, we have $e(u) = e(v) = dist(u,v) = diam(G)-1$. Let $x,y$ be the two ends of a dominating diametral path $P$, that exists by Lemma~\ref{lem:dom-dp}. We pick $u^* \in V(P) \cap N[u]$ and $v^* \in V(P) \cap N[u]$, that exist because $P$ is dominating. W.l.o.g., $dist(x,u^*) \leq dist(x,v^*)$. Then:
\begin{align*}
diam(G) &= dist(x,y) = dist(x,u^*) + dist(u^*,v^*) + dist(v^*,y) \\
&\geq (dist(x,u)-1) + (dist(u,v)-2) + (dist(v,y) - 1) \\
&= (diam(G)-2) + (diam(G)-3) + (diam(G)-2) = 3 \cdot diam(G) - 7
\end{align*}
where the equalities of the last line follow from Lemma~\ref{lem:lexbfs-chordal}.
It implies $diam(G) \leq 3$. Therefore, if $e(v) \geq 3$, we cannot have $e(v) < diam(G)$.

\smallskip
\noindent
For the second part of the theorem, we essentially rely on a previous observation from~\cite{DHV19}. The bichromatic diameter problem consists in, given a graph and two vertex-subsets $A$ and $B$, to compute the maximum distance between any vertex of $A$ and any vertex of $B$. Under SETH, we cannot compute the bichromatic diameter of split graphs~\cite{BCH16}. Now, given a split graph $G$ with clique $K$ along with two subsets $A,B$ in its stable set, we can add in linear time fresh new vertices $a,b \notin V(G)$ s.t. $N[a] = A \cup K \cup \{a,b\}$ and $N[b] = B \cup K \cup \{a,b\}$. In doing so, we get a new split graph $G'$, for which solving the diameter is equivalent to compute the bichromatic diameter of $G$. Note that by construction, $ab$ is a dominating edge of $G'$. As a result, already for split graphs with a dominating edge, we cannot decide whether the diameter is either two or three.
\end{proof}

\subsection{Chordal dominating pair graphs}\label{sec:dompair}

We complete the results of Sec~\ref{sec:duality} by showing that the stronger property of having a dominating pair {\em for each induced subgraph} implies a truly subquadratic algorithm for diameter computation. 

\begin{theorem}\label{thm:dom-pair-chordal}
There is a truly subquadratic algorithm for computing the diameter of chordal dominating pair graphs.
\end{theorem}

Note that in contrast to this above Theorem~\ref{thm:dom-pair-chordal}, the dichotomy result of Theorem~\ref{thm:dicho} also applies to the chordal graphs with a dominating pair ({\it a.k.a.}, chordal weakly dominating pair graphs). This is because having a dominating shortest-path and a dominating edge are a weaker and a stronger property than having a dominating pair, respectively. 

\smallskip
\noindent
The remaining of this subsection is devoted to the proof of Theorem~\ref{thm:dom-pair-chordal}.
Note that according to Theorem~\ref{thm:dicho}, the only difficulty is in order to decide whether the diameter is either two or three.
If we further use Theorem~\ref{thm:framework} (at the price of having a randomized algorithm), then we are left solving {\sc Split-OV} for the split dominating pair graphs.
The following result is an easy corollary of the characterization proven in~\cite{DeK02}:

\begin{lemma}[~\cite{DeK02}]\label{lem:h-free}
A split graph is a dominating pair graph if and only if it is $B_1$-free, where $B_1$ is the graph of Fig~\ref{fig:b1}.
\end{lemma}

\begin{figure}[!h]
\begin{center}
\begin{tikzpicture}
\draw node[scale=.5,fill,circle] at (-1,0) {};
\draw node[scale=.5,fill,circle] at (1,0) {};
\draw node[scale=.5,fill,circle] at (0,1) {};
\draw node[scale=.5,fill,circle] at (-2,0) {};
\draw node[scale=.5,fill,circle] at (2,0) {};
\draw node[scale=.5,fill,circle] at (0,2) {};
\draw[thick] (-1,0) -- (0,1) -- (1,0) -- (-1,0) -- (-2,0);
\draw[thick] (1,0) -- (2,0); 
\draw[thick] (0,1) -- (0,2); 
\end{tikzpicture}
\end{center}
\caption{Forbidden induced subgraph $B_1$.}
\label{fig:b1}
\end{figure}
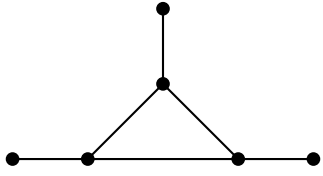

Our main technical contribution in this subsection is as follows:

\begin{lemma}\label{lem:vc-split}
Every $B_1$-free split graph has VC-dimension at most $3$.
\end{lemma}

\begin{proof}
Let $G = (K \cup S, E)$ be a split graph with clique $K$ and stable set $S$. 
Suppose for the sake of contradiction $G$ is $B_1$-free and there exists a $X \subseteq K \cup S, \ |X| \geq 4$, that is shattered. 

We first prove that either $X \subseteq K$ or $X \subseteq S$.
Indeed, by contradiction let $u \in X \cap K, \ v \in X \cap S$.
Since $X$ is shattered, there exists a $z$ s.t.  $N[z] \cap X = \{u,v\}$.
Furthermore, since $G$ is a split graph and $v \in S$, either $z = v$ or $z \in K$.
But if $z = v$, then, for any $z' \in N[z]$ we have $u \in N[z']$, and so there can be no such $z'$ s.t. $N[z'] \cap X = \{v\}$.
Thus, necessarily, $z \in K$. It implies $X \setminus \{u,v\} \subseteq S$ (otherwise, since $K \cap X \subseteq N[z]$, we could not have $N[z] \cap X = \{u,v\}$, a contradiction).
However, let $w \in X \setminus \{u,v\}$.
Again, since $X$ is shattered, there exists a $z'$ s.t. $N[z'] \cap X = \{v,w\}$.
But necessarily, $z' \in K$, and therefore we also have $u \in N[z']$, a contradiction.
The latter proves, as claimed, either $X \subseteq K$ or $X \subseteq S$.
Note that the above still applies if $|X| \geq 3$, and that we did not use in our proof the fact that $G$ is $B_1$-free.

If $X \subseteq K$ then, let $u,v,w \in X$. There exist $z_u,z_v,z_w$ s.t. $N[z_u] \cap X = \{u\}, \ N[z_v] \cap X = \{v\}, \ N[z_w] \cap X = \{w\}$. Necessarily, $z_u,z_v,z_w \in S$. But then, $u,v,w,z_u,z_v,z_w$ induce a copy of $B_1$, that is a contradiction. 
Conversely, if $X \subseteq S$ then, let $u,v,w,r \in X$. There exist $z_u,z_v,z_w$ s.t. $N[z_u] \cap X = \{u,r\}, \ N[z_v] \cap X = \{v,r\}, \ N[z_w] \cap X = \{w,r\}$. Necessarily, $z_u,z_v,z_w \in K$. But then, $u,v,w,z_u,z_v,z_w$ induce a copy of $B_1$, that is a contradiction.  
\end{proof}

This above Lemma~\ref{lem:vc-split} gives us the opportunity to use powerful techniques from previous work in order to solve {\sc Split-OV}. 

\begin{theorem}[special case of Theorem $1$ in~\cite{DHV20}]\label{thm:split-ov-vc-dim}
For every $d > 0$, there exists a constant $\varepsilon_d \in (0,1)$ s.t. for any $(G,S_A,S_B)$ where $G$ has VC-dimension at most $d$, we can solve {\sc Split-OV} in deterministic $\tilde{\cal O}(\ell n^{1-\varepsilon_d})$ time.
\end{theorem}

Some additional remarks are needed.
First of all, Theorem $1$ in~\cite{DHV20} addresses a monochromatic variant of {\sc Disjoint Sets}, that is slightly different than OV. A simple trick (presented in the proof of Theorem $10$ in~\cite{DuD19+}) allows us to reduce {\sc Split-OV} to this monochromatic variant, up to increasing the VC-dimension from $d$ to some value in ${\cal O}(d\log{d})$. Second, the time complexity for Theorem~\ref{thm:split-ov-vc-dim} is in $\tilde{\cal O}(\ell n^{1-\varepsilon_d})$, and not in ${\cal O}(\ell^b)$ for some $b \geq 1$ as it was stated in Theorem~\ref{thm:framework}. However, since $\ell n^{1-\varepsilon_d} = {\cal O}(\ell^{2-\varepsilon_d})$, Theorem~\ref{thm:framework} can still be applied in this case.\qed 

\section{Dominating triples}

A {\em dominating target} for a graph $G$ is a subset of vertices $D$ s.t. every connected graph containing all of $D$ is a dominating set. In particular, dominating pairs are exactly dominating targets of cardinality two. An interesting generalization of Lemma~\ref{lem:dominating} is that every graph of asteroidal number at most $k$ contains a dominating target of cardinality at most $k$~\cite{KKM01}. In this last section, we study chordal graphs with a dominating triple (dominating set of cardinality three).

\begin{theorem}\label{thm:dom-triple}
For every chordal graph with a dominating triple, if the algorithm $3$-sweep LexBFS outputs a vertex of eccentricity $d \geq 10$, then it is the diameter. In particular, there is a linear-time algorithm for deciding whether the diameter is at least $10$ on this graph class.
\end{theorem} 

We use in our proof a few results from Metric Graph Theory, that we now introduce.
A {\em geodesic triangle} with corners $x,y,z \in V$, denoted in what follows by $\Delta(x,y,z)$, is the union $P(x,y) \cup P(y,z) \cup P(z,x)$ of three shortest-paths connecting its corners. The three of $P(x,y),P(y,z) ,P(z,x)$ are called the sides of the triangle. We say that $\Delta(x,y,z)$ is $\delta$-slim if the maximum distance between any vertex on one side $P(x,y)$ and the other two sides $P(y,z) \cup P(z,x)$ is at most $\delta$. A graph is called $\delta$-slim if all its geodesic triangles are.

\begin{lemma}[\cite{MoD19}]\label{lem:slim-chordal}
Every chordal graph is $1$-slim.
\end{lemma}

Three vertices $x,y,z$ form a {\em metric triangle} if, for any choice of shortest-paths $P(x,y), P(y,z), P(z,x)$ connecting them, the latter can only pairwise intersect at their endpoints. If furthermore, $dist(x,y) = dist(y,z) = dist(z,x) = k$, then we say of this metric triangle that it is equilateral of size $k$. In what follows, recall that the meshed graphs are a superclass of chordal graphs.

\begin{lemma}[\cite{BaC02}]\label{lem:equilateral}
Every metric triangle in a meshed graph $G$ is equilateral.
\end{lemma}

\begin{corollary}\label{cor:2-eq}
Every metric triangle in a chordal graph $G$ is equilateral of size at most two.
\end{corollary}

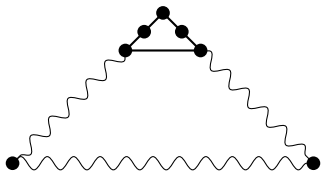
\begin{figure}[!h]
\begin{center}
\begin{tikzpicture}
\draw node[scale=.5,fill,circle] at (-2,-1) {};
\draw node[scale=.5,fill,circle] at (2,-1) {};
\draw node[scale=.5,fill,circle] at (0,1) {};
\draw node[scale=.5,fill,circle] at (-.25,.75) {};
\draw node[scale=.5,fill,circle] at (.25,.75) {};
\draw node[scale=.5,fill,circle] at (-.5,.5) {};
\draw node[scale=.5,fill,circle] at (.5,.5) {};
\draw[thick] (-.5,.5) -- (-.25,.75) -- (0,1) -- (.25,.75) -- (.5,.5);
\draw[thick] (-.5,.5) -- (.5,.5);
\draw[snake it] (-.5,.5) -- (-2,-1);
\draw[snake it] (.5,.5) -- (2,-1);
\draw[snake it] (-2,-1) -- (2,-1);
\end{tikzpicture}
\end{center}
\caption{To the proof of Corollary~\ref{cor:2-eq}.}
\label{fig:eq}
\end{figure}

\begin{proof}
Suppose for the sake of contradiction that there exists a metric triangle $x,y,z$ of size at least three.
By Lemma~\ref{lem:equilateral}, this triangle is equilateral of size $k \geq 3$.
Fix three shortest-paths $P(x,y), P(y,z), P(z,x)$ connecting these vertices, thus obtaining a geodesic triangle $\Delta(x,y,z)$. Let $u \in P(x,y)$ s.t. $dist(x,u) = 2$.
By Lemma~\ref{lem:slim-chordal}, there exists a neighbour $v \in N(u) \cap \left( P(y,z) \cup P(z,x) \right)$. We need to consider two cases.
\begin{itemize}
\item {\it Case $v \in P(z,x)$.} Then, $dist(x,v) \geq 2$ because otherwise, there would be a shortest $xy$-path going through the edge $uv$, thus contradicting that $x,y,z$ is a metric triangle. In the same way, $dist(z,v) \geq k-2$, and therefore, we have: $dist(x,v) = 2, \ dist(z,v) = k-2$ (see Fig.~\ref{fig:eq}).
Let $[x,s,u]$ and $[x,t,v]$ be shortest subpaths of $P(x,y)$ and $P(z,x)$, respectively.
Since $[x,s,u,v,t,x]$ is a cycle of length five, there exists a chord.
However, the only possible such chord is $st$ ({\it i.e.}, because $ut \in E$ and $vs \in E$ imply $t$ is on a shortest $xy$-path and $s$ is on a shortest $xz$-path, respectively). In this situation, $u,v,s,t$ induces a cycle of length four, a contradiction.
\item {\it Case $v \in P(y,z)$.} We prove next that $k \leq 4$. Indeed, this is the case if $v \in \{y,z\}$. Otherwise, since $x,y,z$ is a metric triangle, we cannot have $v$ on a shortest $xy$-path nor on a shortest $xz$-path. Then, $dist(y,v) \geq k-2, \ dist(z,v) \geq k-2$, and therefore, $k = dist(y,z) = dist(y,v) + dist(v,z) \geq 2k-4$. It implies $k \leq 4$. If $k = 4$ then, since we also have $dist(y,u) = 2$, we are back to the previous case up to replacing $x$ by $y$. From now on, let us assume $k=3$. Since the geodesic triangle considered is a cycle of length nine, there exists a chord. However, consider any such chord $st$. Vertices $s$ and $t$ are adjacent to some vertices amongst $\{x,y,z\}$, and so, the edge $st$ is on a path of length three between two corners of the triangle. Since $x,y,z$ is a metric triangle, the only possible chords are, for every of $x,y,z$, between their two neighbours on different sides of the triangle. However, in this situation we are left with an induced cycle of length at least six, thus contradicting that $G$ is chordal.
\end{itemize}
\end{proof}
Note that the bound of Corollary~\ref{cor:2-eq} is sharp, as it is shown, {\it e.g.}, by the $3$-sun.

\smallskip
\noindent
Finally, a {\em quasi-median} for $x,y,z \in V$ is a triple $x^*,y^*,z^* \in V$ s.t.: 
$$\begin{cases}
dist(x,y) = dist(x,x^*) + dist(x^*,y^*) + dist(y^*,y) \\
dist(y,z) = dist(y,y^*) + dist(y^*,z^*) + dist(z^*,z) \\
dist(z,x) = dist(z,z^*) + dist(z^*,x^*) + dist(x^*,x).
\end{cases}$$
Chalopin et al.~\cite{CCPP14} observed that every triple of vertices has a quasi median which is a metric triangle. 

\begin{proof}[Proof of Theorem~\ref{thm:dom-triple}]
Let $x,y,z$ be a dominating triple, and let $x^*,y^*,z^*$ be a corresponding pseudo median that is also a metric triangle. We fix shortest-paths $P(x,x^*), P(y,y^*), P(z,z,^*)$ along with a geodesic triangle $\Delta(x^*,y^*,z^*)$ with sides $P(x^*,y^*), P(y^*,z^*), P(z^*,x^*)$.
In doing so, since  $x^*,y^*,z^*$ is a pseudo median, we also get shortest-paths $P(x,y),P(y,z),P(z,x)$.
For instance, $P(x,y) = P(x,x^*) \cup P(x^*,y^*) \cup P(y^*,y)$.
Furthermore, since $x,y,z$ is a dominating triple, the union $H$ of these above shortest-paths is a dominating set of $G$.
Let $u$ be last visited by a LexBFS, and let $v$ be last visited by $LexBFS(u)$.
Finally, let $(s,t)$ be a diametral pair of $G$.
Assume toward a contradiction $diam(G) \geq 11$ and $e(v) < diam(G)$ (we will explain at the end of the proof how to lower the bound on the diameter to $10$).

We first prove as an intermediate claim that for some two vertices $\alpha,\beta \in \{x,y,z\}$, we have that $dist(u,P(\alpha,\beta)) + dist(v,P(\alpha,\beta)) \leq 3$. Indeed, let $u^* \in N[u] \cap H$ and $v^* \in N[v] \cap H$. If $u^*,v^* \in P(\alpha,\beta)$ for some $\alpha,\beta \in \{x,y,z\}$, then we get $dist(u,P(\alpha,\beta)) + dist(v,P(\alpha,\beta)) \leq 2$. From now on, we assume that it is not the case. If $u^*,v^* \in \Delta(x^*,y^*,z^*)$, then since by Corollary~\ref{cor:2-eq} this metric triangle is equilateral of size $\leq 2$, $dist(u^*,v^*) \leq 3$. However, it implies by Lemma~\ref{lem:lexbfs-chordal} that $diam(G) \leq dist(u,v) + 1 \leq dist(u^*,v^*) + 3 \leq 6$, a contradiction. Hence, let us assume for instance that $u^* \in P(x,x^*)$ but $v^* \in P(y^*,z^*)$ is on the only side of $\Delta(x^*,y^*,z^*)$ that is not on a shortest-path between $x$ and one of $\{y,z\}$ (all other cases are symmetrical to this one). Since $x^*,y^*,z^*$ is equilateral of size at most two, for {\em any} $\beta \in \{y,z\}$ we get $dist(v^*,P(x^*,\beta^*)) = 1$. As a result, $dist(u,P(x,\beta)) + dist(v,P(x,\beta)) \leq 1 + 2 = 3$, thus proving the claim. Furthermore, as a by-product of our proof, we also have in this case $\max\{dist(u,P(\alpha,\beta)),dist(v,P(\alpha,\beta))\} \leq 2$.

We prove as another intermediate claim that for some $w \in \{s,t\}$ we have $dist(w,P(\alpha,\beta)) \leq 2$.
Indeed, let $s^* \in N[s] \cap H$ and $t^* \in N[t] \cap H$. 
If one amongst $s^*$ and $t^*$ is a vertex of $\Delta(x^*,y^*,z^*)$, say it is the case of $s^*$, then since the latter triangle is equilateral of size at most two, $dist(s,P(\alpha,\beta)) \leq 1 + dist(s^*,P(\alpha^*,\beta^*)) \leq 2$. From now on, we assume that $s^*$ and $t^*$ are not vertices of $\Delta(x^*,y^*,z^*)$. If we write $\{\alpha,\beta,\gamma\} = \{x,y,z\}$, then we may assume w.l.o.g. $s^*,t^* \in P(\gamma,\gamma^*)$ (otherwise, one amongst $s^*,t^*$ is a vertex of $P(\alpha,\beta)$, and so we are done). We prove as a subclaim that $dist(u^*,P(\alpha,\gamma) \cup P(\beta,\gamma)) \leq 1$. Indeed, to see that, it suffices to recall that either $u^* \in P(\alpha,\alpha^*) \cup P(\beta,\beta^*)$ or $u^* \in P(\alpha^*,\beta^*)$. In the former case, $dist(u^*,P(\alpha,\gamma) \cup P(\beta,\gamma)) = 0$, while in the latter case, since we have a metric triangle that is equilateral of size two, $dist(u^*,P(\alpha^*,\gamma^*)) = dist(u^*,P(\beta^*,\gamma^*)) = 1$. Thus, choosing $u' \in P(\alpha,\gamma) \cup P(\beta,\gamma)$ at minimum distance from $u^*$:
\begin{align*}
dist(s^*,u^*) &\leq dist(u^*,u') + (dist(\gamma,u') - dist(\gamma,s^*))\\ 
&\leq 1 + diam(G) - dist(\gamma,s^*).
\end{align*} 
In particular, $dist(s,u) \leq dist(s^*,u^*) + 2 \leq  diam(G) + 3 - dist(\gamma,s^*)$. Since we have $dist(s,u) = diam(G) - 1$ by Lemma~\ref{lem:lexbfs-chordal}, it follows from the above inequalities that we have $dist(\gamma,s^*) \leq 4$. We get in the exact same way $dist(\gamma,t^*) \leq 4$. However, it implies $diam(G) = dist(s,t) \leq 2 + dist(s^*,t^*) \leq 2 + dist(s^*,\gamma) + dist(\gamma,t^*) \leq 10$, a contradiction.

Overall, let $u',v',w' \in P(\alpha,\beta)$ at minimum distance from $u,v,w$ respectively. By Lemma~\ref{lem:lexbfs-chordal}, $u,v,w$ are pairwise at distance $diam(G)-1$. Let $\{r_1,r_2,r_3\} = \{u,v,w\}$ s.t. $r_2'$ is (metrically) between $r_1'$ and $r_3'$ onto the shortest-path $P(\alpha,\beta)$. Then:
\begin{align*}
diam(G) &\geq dist(\alpha,\beta) \\
&\geq dist(r_1',r_2') + dist(r_2',r_3') \\
&\geq dist(r_1,r_2) + dist(r_2,r_3) - dist(r_1,r_1') - 2 \cdot dist(r_2,r_2') - dist(r_3,r_3') \\
&\geq 2(diam(G)-1) - dist(u,u') - dist(v,v') - dist(w,w') \\
&- \max\{dist(u,u'), dist(v,v'), dist(w,w')\} \\
&\geq 2 \cdot diam(G) - 2 - 3 - 2 - 2 = 2 \cdot diam(G) - 9.
\end{align*} 
But then, $diam(G) \leq 9$, a contradiction.

Finally, let us sketch how we can decide in linear time whether the diameter is at least equal to $d$, for any $d \geq 10$.
For $d=10$, we claim that it suffices to execute a $2$-sweep LexBFS. Indeed, by Lemma~\ref{lem:lexbfs-chordal}, the outputted vertex has eccentricity either $9$ or $10$, and if it is $9$ then, by Corollary~\ref{cor:lexbfs-chordal}, we have $diam(G) = 9 < 10$. Otherwise, $d \geq 11$ and we apply the algorithm $3$-sweep LexBFS, whose correctness follows from the above analysis. 
\end{proof}

We left open the following intriguing question. For any $k \geq 2$, does there exist a $d_k > 0$ s.t., for any (chordal) graph with a dominating target of cardinality at most $k$, we can decide in truly subquadratic time whether the diameter is at least $d_k$, and if so compute the diameter exactly?

%\section{Beyond chordal graphs}\label{sec:gal-graphs}
%
%\subsection{}
%
%\begin{theorem}\label{thm:spd}
%\end{theorem}
%
%mention that diam must be at least five due to the counter-example of Corneil et al.
%
%\subsection{}
%
%\begin{theorem}\label{thm:dt}
%\end{theorem}

\bibliographystyle{abbrv}
\bibliography{biblio-AT}

\end{document}